\documentclass[letterpaper, 10 pt, conference]{ieeeconf}  

\IEEEoverridecommandlockouts                              
\title{\LARGE \bf
Feedback Linearizable Discretizations of Second Order Mechanical Systems using Retraction Maps
}

\author{Shreyas N B$^{1}$, David Mart{\'\i}n Diego$^{2}$ and Ravi Banavar$^{3}$
\thanks{*David\ Mart\'{\i}n de Diego acknowledges financial support from the Spanish Ministry of Science and Innovation, under grants  PID2022-137909NB-C21, TED2021-129455B-I00 and  CEX2019-000904-S funded by MCIN/AEI\-/10.13039\-/501100011033.}
\thanks{$^{1,3}$
        Systems and Control Engineering, Indian Institute of Technology Bombay, Mumbai, 400076, Maharashtra, India  $^1${\tt\small shreyasnb@iitb.ac.in}, $^3${\tt\small banavar@iitb.ac.in} }
\thanks{$^{2}$ Instituto de Ciencias Matem\'aticas (CSIC-UAM-UC3M-UCM), Calle Nicol\'as Cabrera 13-15, 28049 Madrid, Spain
{\tt\small david.martin@icmat.es}}
}
\usepackage{mymacros}

\begin{document}

\maketitle
\thispagestyle{empty}
\pagestyle{empty}


\begin{abstract}

Mechanical systems are most often described by a set of continuous-time, nonlinear, second-order differential equations (SODEs) of a particular structure governed by the covariant derivative. The digital implementation of controllers for such systems requires a discrete model of the system and hence requires numerical discretization schemes. Feedback linearizability of such sampled systems, however, depends on the discretization scheme employed. In this article, we utilize retraction maps and their lifts to construct feedback linearizable discretizations for SODEs which can be applied to many mechanical systems.

\end{abstract}


\section{Introduction}

The design of control laws for mechanical systems has evolved rapidly in the last century. Digital controllers facilitate the implementation of these algorithms in various ways, a popular way being discretization schemes that evolve numerically with the continuous-time dynamics. While dealing with systems on Euclidean spaces, we could use Runge-Kutta, Euler discretization, etc., these discretization schemes, when applied to general manifolds, do not assure that the system states stay on the manifold. In other words, we must know how to move on the manifold. For example, on the Riemannian manifold $(M, g)$, this notion would be described by the exponential map $\exp_x: T_x M \mapsto M$ at the point $x \in M$. Thus, there is a need for a universal tool that allows us to construct integrators that respect the underlying geometry of general manifolds. Retraction and discretization maps generalize Euclidean discretizations on general manifolds \cite{21MBLDMdD}. 

Another aspect related to the control of mechanical systems that has been under study for some time is transforming nonlinear systems to a locally linear system and an invertible control law through a coordinate transformation. Feedback linearization allows us to do this, enabling the utilization of standard powerful methods that can be applied to linear systems like pole placement to design control laws for nonlinear systems. This has led to the work on constructing feedback linearizable discretizations for first-order nonlinear systems (refer \cite{ashutosh}). 
Although feedback linearization for continuous-time systems has been developed well in \cite{brockett1978feedback}, \cite{jacubczyk1980linearization}, \cite{isidori1982feedback}, when addressing mechanical systems of the second-order, we must ensure that the coordinate transformation yields a linear system which is also mechanically equivalent to the original mechanical system, i.e., we must design feedback which linearizes the mechanical system, while simultaneously preserving its mechanical structure (see \cite{10076262},  \cite{nowicki}, \cite{NOWICKI2023111098}, \cite{RESPONDEK2012102}).\\
\textbf{Contribution}: In this article, given a mechanical system that falls under a specific class, we utilize lifts of retraction maps to construct discretizations that are feedback linearizable. We propose a structured discretization scheme that can be applied to second-order mechanical systems. \\
\textbf{Organization}: The article is organized as follows: Section \ref{sec:mech} introduces us to the field, with the definition of a mechanical control system, as defined in \cite{10076262},\cite{RESPONDEK2012102}. In Section \ref{sec:fl}, we look at the feedback linearizability of such mechanical systems in continuous time, where we define the existence of a specific class of systems that are feedback linearizable while simultaneously preserving its mechanical structure. In Section \ref{sec:disc}, we discuss discretization maps and their tangent lifts to construct integrators on the double-tangent bundle $TTM$ for second-order systems. We present our main result in Section \ref{sec:fl_disc}, where we utilize the tangent lifts defined in Section \ref{sec:disc} to construct a feedback linearized discrete system, and we demonstrate these results on a simple mechanical system in Section \ref{sec:example}. 


\subsection{Previous Work}

\subsubsection{Retraction and Discretization Maps} 

\begin{defn}
\label{eq:ret}
    Let $M$ be an $n$ dimensional manifold, and $TM$ be its tangent bundle. Denote the canonical projection on the manifold to be $\tau_M(x,y) = x$, where $\tau_M: TM \rightarrow M$. We define a \textbf{retraction map} on a manifold $M$ as a smooth map $R: TM \to M$, such that if $R_x$ be the restriction of $R$ to $T_x M$, then the following properties are satisfied:

    \begin{enumerate}
        \item $R_x (0_x) = x$ where $0_x$ is the zero element of $T_x M$.
        \item $\text{D}R_x (0_x ) = T_{0_x} R_x = \text{Id}_{T_x M} $, where $\text{Id}_{T_x M}$ denotes the identity mapping on $T_x M$.
    \end{enumerate}
    Define an open neighborhood $U \subset TM$ around the zero section of the tangent bundle. If $(x,y) \in U$, where $y\in T_xM$,  then $R_d: U \rightarrow M \times M$ is called a \textbf{discretization map}, if it satisfies the following properties:
\begin{enumerate}
    \item $R_d(x, 0_x) = (x,x)$
    \item $T_{(x,0_x)}R^2 - T_{(x, 0_x)}R^1 = \text{Id}_{T_x M}$, which is the identity map on $T_x M$ for any $x \in M$.
\end{enumerate}
\end{defn}
We note that here, we define $R_d(x, v) = (R^1_x(v), R^2_x(v))$, and $T_{(x,0_x)}R^a = \text{D}R^a(x,0_x)$ for $a=1,2$.
Consequently, it can be shown that any discretization map $R$ is a local diffeomorphism.

Given a vector field $X \in \mathfrak{X}(M)$ on $M$ and a discretization map $R_d$ and a fixed time discretization map $t \mapsto (t-\alpha h, t+(1-\alpha)h), \ \alpha \in [0,1]$, the discretization of $X$ (for a small enough step-size $h > 0$) is defined by:
\begin{equation}
    R_d^{-1} (x_k, x_{k+1}) = h X(\tau_M(R_d^{-1}(x_k, x_{k+1})))
\end{equation}

Let us consider an example where $R_d(x, v) = (x, x+v)$. If we have $\dot{x} = X(x(t))$, then this choice of discretization map yields the scheme $x_{k+1} = x_k + hX(x_k)$ where $x(t_k)=x_k$. This is the standard \textit{forward Euler} scheme. Similarly $R_d(x,v) = (x-v, x)$ yields the \textit{backward Euler} method.

\begin{rmk}
More generally if we fix  a discretization  map on the space ${\mathbb R}\times M$ given by $\tilde{R}_d: T({\mathbb R}\times M)\rightarrow ({\mathbb R}\times M)\times ({\mathbb R}\times M)$ and a time-dependent dynamics given by $\dot{x}=X(t, x)$, we define the corresponding discretization as 
\begin{align*}
 &(\hbox{pr}_{TM}\circ \tilde{R}_d^{-1}) (t_k, x_k, t_{k+1}, x_{k+1})\\
 &= h X(\tau_{{\mathbb R}\times M}(\tilde{R}_d^{-1}(t_k, x_k, t_{k+1}, x_{k+1})))\\
& h= (\hbox{pr}_{2,\mathbb R}\circ \tilde{R}_d^{-1}) (t_k, x_k, t_{k+1}, x_{k+1})
\end{align*}
where $\hbox{pr}_{TM}: T({\mathbb R}\times M)\rightarrow TM$ and 
$\hbox{pr}_{T\mathbb R}: T({\mathbb R}\times M)\rightarrow T{\mathbb R}$ are the canonical projections and $\hbox{pr}_{2,\mathbb R}=\hbox{pr}_2\circ \hbox{pr}_{T\mathbb R}$.

Moreover, these extended retractions allow us to introduce discrete time reparametrizations (discrete Sundman transformations) adding to the dynamics given by $\dfrac{dx}{dt}=X(t, x)$ the time transformation $\dfrac{d\tau}{dt}=\dfrac{1}{f(x)}$ where $f: M\rightarrow {\mathbb R}$ verifies $f>0$ (see \cite{Cari} and references therein).
\end{rmk}

\subsubsection{Feedback Linearization}
Let $M$ be an $n$-dimensional manifold and $U \in \R[n]$ be open subset representing the control variables. For $u \in U$, let $X(\cdot,u) \in \mathfrak{X}(M)$ be a vector field on $M$. Then, for a fixed time $T > 0$, a continuous-time dynamical system on $M$ is given by
\begin{equation}
\label{eq:1}
    \frac{d}{dt}x(t) = X(x(t), u(t)) \ \forall t \in [0, T] 
\end{equation}
Let $M$ and $N$ be two $n$-dimensional manifolds and $\phi: M \mapsto N$ be a diffeomorphism. If $X \in \mathfrak{X}(M)$ is a vector field on $M$, then $X_{\phi}$ is a vector field on $N$ such that $X_{\phi} := T\phi \circ X \circ \phi^{-1}$, which yields the following mechanical system on $N$:
\begin{equation}
    \dfrac{d}{dt}\tilde{x}(t) = X_{\phi} (\tilde{x}(t), \tilde{u}(t))
\end{equation}

with $\tilde{x}(0) = \phi(x(0))$ and $\tilde{x}(t) = \phi(x(t))$ for all $t \in [0, T]$.
\begin{defn}
    A given system \ref{eq:1} is said to be locally feedback linearizable around a specified point if there exist matrices $A \in \R[n \times n]$ and $B \in \R[n \times m]$ such that $X_{\phi}(\tilde{x}, v) = A\tilde{x} + Bv$ where $v = \psi(\phi^{-1}(\tilde{x}), u)$, $\psi$ being the linearizing feedback, and the corresponding feedback linearized dynamical system is given by:
    \begin{equation}
        \dfrac{d}{dt} \tilde{x}(t) = A \tilde{x}(t) + B v(t) \ \forall t \in [0, T]
    \end{equation}
\end{defn}

\vspace{10pt}
For detailed proofs and background on feedback linearization, we refer the reader to \cite{ashutosh} -- \cite{isidori1982feedback}.

\subsubsection{Constructing feedback linearizable discretization maps}

From \cite{ashutosh}, we thus formulate the construction of feedback linearizable discretization maps, by lifting the discretization map $R_d$ on $TN$ to $R_{d,\phi}$ on $TM$ as shown in Figure \ref{fig:commutator}.
\begin{figure}[h]
    \centering
    \begin{tikzpicture}
  \matrix (m) [matrix of math nodes,row sep=3em,column sep=4em,minimum width=2em]
  {
     TM & TN \\ M \times M & N \times N \\};
  \path[-stealth]
    (m-1-1) edge node [above] {$T\phi$} (m-1-2)
    (m-1-1) edge node [left] {$R_{d,\phi}$} (m-2-1)
    (m-1-2) edge node [right] {$R_d$} (m-2-2)
    (m-2-1) edge node [below] {$\phi \times \phi$} (m-2-2);
\end{tikzpicture}
    \caption{$R_d$ and $R_{d,\phi}$ commute as shown}
    \label{fig:commutator}
\end{figure}
The key concept from \cite{ashutosh}, is the following:
\begin{prop}
    Let $\phi: M \lra N:= \R[n] $ be the linearizing coordinate change, and $\psi: M \times U \lra \R[m]$ be the linearizing feedback, where $U \subset \R[m]$ is the control space. Let $R_d$ be a discretization map on $N$ that discretizes the continuous-time linear system (CTLS) to a discrete-time linear system (DTLS). Then,
    \begin{equation}
        R_{d,\phi} = (\phi \times \phi)^{-1}  \circ R_d \circ T \phi
    \end{equation}
    is a discretization map on $M$, which discretizes the continuous-time system (CTS) to a feedback linearizable discrete-time system (DTS). Here, the linearizing coordinate is given by $x_k = \phi^{-1}(\tilde{x}_k)$ and the auxiliary control $u_k$ is given by $\tilde{u}_k = \psi(x_k, u_k)$ for the DTS. 
\end{prop}


\section{Mechanical Control Systems}
\label{sec:mech}

We define a mechanical control system and its feedback linearization as proposed in \cite{10076262} and \cite{NOWICKI2023111098}.

\begin{defn}
    A mechanical control system $(\mathcal{MS})_{(n,m)}$ is defined by a $4$-tuple $(M, \nabla, \mathfrak{g}, e)$ where:
    \begin{itemize}
        \item $M$ is an $n$-dimensional manifold
        \item $\nabla$ is a symmetric affine connection on $M$
        \item $\mathfrak{g} = \{g_1, \dots, g_m\}$ is an $m$-tuple of control vector fields on $M$
        \item $e$ is an uncontrolled vector field on $M$
    \end{itemize}
    $(\mathcal{MS})_{(n,m)}$ can be represented by the differential equation:
    \begin{equation}
        \nabla_{\dot{x}} \dot{x} = e(x) + \sum_{r=1}^m g_r(x) u_r 
    \end{equation}
    Or equivalently in local coordinates $x = (x^1, \dots, x^n)$ on $M$, 
    \begin{equation}\label{SODE-initial}
        \ddot{x}^i = - \Gamma ^i_{jk}(x)\dot{x}^j \dot{x}^k + e^i(x) + \sum_{r=1}^m g^i_r(x)u_r
    \end{equation}
    If we write this as two first-order differential equations:
    \begin{equation}\label{SODE-nonlinear}
        \dot{x}^i  = y^i; \
        \dot{y}^i  = - \Gamma^i_{jk}(x)y^jy^k + e^i(x) + \sum_{r=1}^m g_r^i(x)u_r
    \end{equation}
\end{defn}


\section{Mechanical Feedback Linearization}
\label{sec:fl}
We consider the problem of bringing a mechanical system $(\mathcal{MS})_{(n,m)}$ into a linear mechanical form through a transformation and mechanical feedback. This \textit{M}echanical \textit{F}eedback linearization (or $MF$-linearization \cite{nowicki}) must preserve the structure of the tangent bundle $TM$

We can see that the diffeomorphism on $M$ induces a mechanical diffeomorphism on $TM$ given by $(\tilde{x}, \tilde{y}) = \left(\phi(x), D\phi(x)y \right)$. The proof of the action of $MF$ preserving the mechanical structure of $(\mathcal{MS})_{(n,m)}$ can be referred to in more detail in \cite{nowicki}.

According to the definition in Appendix B. [\ref{appendix-b} \ref{eq:mf_diff}], the $MF$-transformations act on the vector fields $g_r$ and $e$ through $(\phi, \alpha, \beta)$. We define the following relevant distributions:
\begin{equation}
    \begin{split}
        \mathcal{E}^0 & = \text{span} \{ g_r, 1 \leq r \leq m \} \\
        \mathcal{E}^j & = \text{span} \{ \text{ad}^i_e g_r, 1 \leq r \leq m, 0 \leq i \leq j \}
    \end{split}
\end{equation}
Thus, we state the following theorem:
\begin{thm}\label{thm:mfl}
    A mechanical system $(\mathcal{MS})_{(n,m)}$ is said to be mechanical feedback ($MF$) linearizable, locally around $x_0 \in M$ if and only if, in the neighborhood of $x_0$, it satisfies the following conditions:

    \begin{enumerate}
        \item $(ML1)$ $\mathcal{E}^0$ and $\mathcal{E}^1$ are of constant rank
        \item $(ML2)$ $\mathcal{E}^0$ is involutive
        \item $(ML3)$ $\text{ann } \mathcal{E}^{0} \subset \text{ann } \mathcal{R}$
        \item $(ML4)$ $\text{ann } \mathcal{E}^0 \subset \text{ann } \nabla g_r \ \forall r: 1 \leq r \leq m$
        \item $(ML5)$ $\text{ann } \mathcal{E}^1 \subset \text{ann } \nabla^2 e$
    \end{enumerate}
\end{thm}
where $\mathcal{R}$ is the Riemannian curvature tensor.


Note that the above conditions are valid without the assumption of controllability of the linearized mechanical system.
 The following proposition is explicitly stated for planar mechanical systems where $n=2$. (refer \cite{nowicki}).

\begin{prop}
\label{prop:planar_mech}
A planar mechanical system $\mathcal{(MS)}_{(2,1)}$ is locally $MF$-linearizable at $x_0 \in M$ to a controllable $\mathcal{(LMS)}_{(2,1)}$, if and only if it satisfies the following conditions:
\begin{enumerate}
    \item $(MD1)$ $g$ and $\text{ ad}_e g$ are independent
    \item $(MD2)$ $\nabla_g g \in \mathcal{E}^0$ and $\nabla_{\text{ad}_e g} g \in \mathcal{E}^0$
    \item $(MD3)$ $\nabla^2_{g, \text{ad}_e g} \text{ad}_e g - \nabla^2_{\text{ad}_e g, g} \text{ad}_e g \in \mathcal{E}^0$
\end{enumerate}
\end{prop}
\vspace{5pt}

\begin{defn}
    A mechanical control system $(\mathcal{MS})_{(n,m)} = (M, \nabla, \mathfrak{g}, e)$ is called $MF$-linearizable if it is $MF$-equivalent to a linear mechanical system $(\mathcal{LMS})_{(n,m)} = (\R[n], \bar{\nabla}, \mathfrak{b}, A \tilde{x})$, where $\bar{\nabla}$ is an affine connection with the Christoffel symbols zero ($\bar{\nabla}$ is a flat connection) and $\mathfrak{b} = \{b_1, \dots, b_m \}$ are constant vector fields. In other words, there exists $(\phi, \alpha, \beta, \gamma) \in MF$ such that 
    \begin{equation}
        \begin{split}
            \phi: M \lra N \ \ \phi(x)  = \tilde{x} \\
            \phi_*\left(\nabla - \sum_{r=1}^m g_r \otimes \gamma^r\right) = \bar{\nabla} \\
            \phi_*\left( \sum_{r=1}^m \beta^r_s g_r \right) = b_s, \ 1 \leq s \leq m \\
            \phi_* \left( e + \sum_{r=1}^m g_r \alpha^r \right) = A \tilde{x}
        \end{split}
    \end{equation}
    Equivalently, we have the corresponding linear mechanical system $(\mathcal{LMS})_{(n,m)}$ as:
    \begin{equation}\label{sode-lin}
            \dot{\tilde{x}}  = \tilde{y}; \ 
            \dot{\tilde{y}}  = A \tilde{x} + \sum_{s=1}^m b_s \tilde{u}_s
    \end{equation}
\end{defn}

\vspace{5pt}

\begin{figure*}
    \centering
    \begin{tikzpicture}[scale=0.8, transform shape]
  \matrix (m) [matrix of math nodes,row sep=3.5em,column sep=5em,minimum width=0.1em]
  {
     TM & TM &  TN & TN \\ TTM & TTM & TTN & TTN \\T(M \times M) & TM \times TM & TN \times TN & T(N \times N) \\ {} & M \times M & N \times N & {} \\};
  \path[-stealth]
    (m-1-2) edge node [above] {$T\phi$} (m-1-3)
    (m-2-2) edge node [above] {$TT\phi$} (m-2-3)
    (m-3-2) edge node [above] {$T\phi \times T\phi$} (m-3-3)
    (m-2-1) edge node [left] {$T\tau_M$} (m-1-1)
    (m-2-2) edge node [left] {$\tau_{TM}$} (m-1-2)
    (m-2-3) edge node [right] {$\tau_{TN}$} (m-1-3)
    (m-2-4) edge node [right] {$T\tau_N$} (m-1-4)
    (m-2-1) edge node [left] {$TR$} (m-3-1)
    (m-2-2) edge node [left] {$R_d^T$} (m-3-2)
    (m-2-3) edge node [right] {$R_{d,\phi}^T$} (m-3-3)
    (m-2-4) edge node [right] {$TR_{d,\phi}$} (m-3-4)
    (m-3-2) edge node [left] {$\tau_{M \times M}$} (m-4-2)
    (m-3-3) edge node [right] {$\tau_{N \times N}$} (m-4-3)
    (m-4-2) edge node [above] {$\phi \times \phi$} (m-4-3);
  \path[<->]
  (m-2-1) edge node [above] {$\kappa_M$} (m-2-2)
  (m-2-3) edge node [above] {$\kappa_N$} (m-2-4);
  \draw[double, double distance=2pt]
    (m-1-1) -- (m-1-2)
    (m-1-3) -- (m-1-4)
    (m-3-1) -- (m-3-2)
    (m-3-3) -- (m-3-4);
  \draw[->, rounded corners]
  (m-1-1) -- node[pos=1,fill=white,inner sep=0pt]{} ++(-3,0) |- node[above right, inner sep=10pt]{$R_d$} (m-4-2);
  \draw[->, rounded corners]
  (m-1-4) -- node[pos=1, fill=white, inner sep=0pt]{} ++ (3,0) |- node[above left, inner sep=10pt]{$R_{d,\phi}$} (m-4-3);
\end{tikzpicture}
    \caption{$R_d^T$ and $R^T_{d,\phi}$ commute as shown. Description in \ref{sec:double-comm}.}
    \label{fig:double-comm}
\end{figure*}

\section{Tangent Lifts of Discretization Maps}
\label{sec:disc}

We have already seen the definition of a discretization map on $M$ given by $R_d: TM \mapsto M \times M$. To perform discretization in second-order differential equations, we require the definition of a discretization map on $TM$, denoted by $R_d^T: TTM \mapsto TM \times TM$, where $TTM$ is the double tangent bundle of $M$ \cite{ABCM2023,21MBLDMdD}.

Let $M$ be an $n$-dimensional manifold and $\tau_M: TM \lra M$ be the canonical projection map of the tangent bundle, and $TTM$ be the double tangent bundle. The manifold $TTM$ intrinsically admits two vector bundle structures such that the first vector bundle structure is canonical with the vector bundle projection $\tau_{TM}: TTM \lra TM$, whereas the second vector bundle structure, the map $T\tau_M: TTM \lra TM$ gives the projection.

Thus, to define $R_d^T$ from a given discretization map $R_d$, we would have to use the canonical involution map $\kappa_M: TTM \lra TTM$, which shows the dual vector bundle structure of $TTM$. Let $(x,y) \in TM$ and $(x,y,\dot{x}, \dot{y}) \in TTM$ be the canonical coordinates. Then:
\begin{equation}
    \kappa_M(x, y, \dot{x}, \dot{y}) = (x, \dot{x}, y, \dot{y})
\end{equation}
It can be shown that $\kappa^2_M = \text{ Id}_{TTM}$ implying $\kappa_M$ is an involution of $TTM$.

\begin{prop}
\label{prop:lift}
    If $R_d$ is a discretization map on $M$, then $R_d^T = TR_d \circ \kappa_M$ is a discretization map on $TM$.
    \newline
    
    \begin{proof}
        For $(x,y,\dot{x}, \dot{y}) \in TTM$, we have that $TR_d(x,y,\dot{x}, \dot{y}) = \left( R_d(x,y), D_{(x,y)} R_d(x,y) (\dot{x}, \dot{y})^T \right)$ and
        $$R_d^T(x, \dot{x}, y, \dot{y}) = (R_d(x,y), D_{(x,y)}R_d(\dot{x}, \dot{y})^T) $$

        Using the properties defined in Definition \ref{eq:ret}, 
        \begin{enumerate}
            \item We know that $R_d(x, 0) = (x,x) \ \forall x \in M$. Thus,
            \begin{equation*}
            \begin{split}
                R_d^T(x,\dot{x}, 0, 0) & = \left( R_d(x, 0), D_{(x,0)} R_d(\dot{x}, 0) \right) \\
                & = (x,x, \dot{x}, \dot{x}) \equiv (x, \dot{x}, x, \dot{x})
            \end{split}
            \end{equation*}
             
            where we trivially identify $T(M \times M) \equiv TM \times TM$.
            \item For this property, we know that
            $$R_d^T(x, \dot{x}, y, \dot{y}) = \left(TR_d^1(x, y; \dot{x}, \dot{y}), TR_d^2(x, y; \dot{x}, \dot{y}) \right)$$
            So, we need to compute 
            $$T_{(0,0)_{(x,\dot{x})}}(TR_d^a)_{(x, \dot{x})}(x, \dot{x}) : T_{(x, \dot{x})}TM \lra T_{(x, \dot{x})}TM$$
            for $a=1,2$, to prove that the map $T(TR_d^2)_{(x,\dot{x})} - T(TR_d^1)_{(x, \dot{x})}$ is the identity map at the zero section $(0,0)_{(x,\dot{x})}$, from $T_{(x, \dot{x})} TM$ to itself.

            We can calculate 
                $$\dfrac{d}{ds}\bigg|_{s=0} \left( R_d^a(x, sy), \partial_{x} R_d^a(x, sy) \dot{x} + \partial_y R_d^a(x, sy) s \dot{y} \right)$$
            
            At $(x, \dot{x}, 0, 0)$, the map $T_{(0,0)_{(x, \dot{x})}}(TR_d^a)_{(x, \dot{x})}$ is thus given by:

            $$\pmat{
            \partial_{y^j}(R_d^a)^i(x, 0) & 0 \\
            \partial_{x^k} \partial_{y^j} (R_d^a)^i(x,0) \dot{x}^k & \partial_{y^j}(R_d^a)^i(x,0)
            }$$

            Thus, using the properties of the discretization map $R_d$, we have the Jacobian matrix of $(TR_d^2)_{(x, \dot{x})} - (TR_d^1)_{(x, \dot{x})}$ at $(0,0)_{(x, \dot{x})}$ as:
            $$\pmat{
            \partial_{y}(R_d^2-R_d^1)(x,0) & 0 \\
            \partial_x(\partial_y(R_d^2 - R_d^1)(x,0))\dot{x} & \partial_y (R_d^2-R_d^1)(x,0)
            }
            $$
            which is indeed equal to the identity $\text{ Id}_{2n \times 2n}$, since $\partial_y (R_d^2-R_d^1)(x,0) = \text{ Id}_{n \times n}$ which also implies $\partial_x(\partial_y(R_d^2 - R_d^1))(x,0) = 0$
        \end{enumerate}
    \end{proof}
\end{prop}

\begin{prop}
\label{prop:lift-commutator}
    Let $M$ and $N$ be $n$ dimensional manifolds and $\phi(x) = \tilde{x}$, where $\phi$ is a diffeomorphism and $x \in M, \tilde{x} \in N$. Let $TM$ and $TN$ be the tangent bundles of $M$ and $N$, respectively. By definition, if $(x,\dot{x}) \in TM$ and $(\tilde{x}, \dot{\tilde{x}}) \in TN$, then $T\phi(x,\dot{x}) = (\tilde{x}, \dot{\tilde{x}})$ through the same diffeomorphism. For a given discretization map $R_d^T$ on $TM$, $R^T_{d,\phi} := (T\phi \times T \phi) \circ R_d^T \circ TT\phi^{-1}$ is a discretization map on $TN$ (refer Figure \ref{fig:double-comm}).
    \newline
    
    \begin{proof}
        For any given $(\tilde{x}, \dot{\tilde{x}}) \in TN$, we have that:
        \begin{equation*}
        \begin{split}
            R^T_{d,\phi}(\tilde{x}, \dot{\tilde{x}}, 0, 0) & = \left( (T\phi \times T \phi) \circ R_d^T \circ TT\phi^{-1} \right) (\tilde{x}, \dot{\tilde{x}}, 0, 0) \\
            & =  (T\phi \times T \phi) \circ R_d^T (x, \dot{x}, 0, 0) \\
            & = (T \phi \times T \phi) (x, \dot{x}, x, \dot{x}) = (\tilde{x}, \dot{\tilde{x}}, \tilde{x}, \dot{\tilde{x}})
        \end{split}
        \end{equation*}
        which proves the first condition in \ref{eq:ret}.
        
        Now, for coordinates $(\tilde{x}, \dot{\tilde{x}}, \tilde{y}, \dot{\tilde{y}}) \in TTN$, 
        \begin{equation*}
            \begin{split}
                & (T_{(\tilde{x}, \dot{\tilde{x}}, 0, 0)} (R^T_{d,\phi})^2 - T_{(\tilde{x}, \dot{\tilde{x}}, 0, 0)} (R^T_{d,\phi})^1)(\tilde{x}, \dot{\tilde{x}}, \tilde{y}, \dot{\tilde{y}}) \\
                & = \dfrac{d}{ds} \bigg|_{s=0} [(T\phi \circ (R_d^T)^1 \circ TT\phi^{-1})(\tilde{x}, \dot{\tilde{x}}, s\tilde{y}, s\dot{\tilde{y}}) \\
                & - (T\phi \circ (R_d^T)^2 \circ TT\phi^{-1})(\tilde{x}, \dot{\tilde{x}}, s\tilde{y}, s\dot{\tilde{y}})] \\
                & = T_{(\tilde{x}, \dot{\tilde{x}})}T\phi \pmat{ \dfrac{d}{ds} \bigg|_{s=0} [(R_d^T)^1(s (TT\phi^{-1})(\tilde{x},\dot{\tilde{x}}, \tilde{y}, \dot{\tilde{y}})) \\
                - (R_d^T)^2(s (TT\phi^{-1})(\tilde{x},\dot{\tilde{x}}, \tilde{y}, \dot{\tilde{y}}))]
                } \\
                & = T_{(\tilde{x}, \dot{\tilde{x}})}T\phi((TT\phi^{-1})(\tilde{x}, \dot{\tilde{x}}, \tilde{y}, \dot{\tilde{y}})) = (\tilde{x}, \dot{\tilde{x}}, \tilde{y}, \dot{\tilde{y}})
            \end{split}
        \end{equation*}
        which proves the second condition in \ref{eq:ret}.

        Thus, using the linearity of the map $TT\phi$, we prove that $R^T_{d,\phi}$ is indeed a discretization map on $TN$.
    \end{proof}
\end{prop}

\subsection{Commutator diagram}
\label{sec:double-comm}
Fig. \ref{fig:double-comm} shows how one can move from one space to another using various maps. The commutator on the double tangent space (inner block) yields to us the relation between the tangent lifts of the discretization maps on $M$ and $N$, which are $R_d^T$ and $R_{d,\phi}^T$, respectively.
 We also see the relation between the canonical isomorphism on the double tangent space of $M$ and $N$, which can be related through the canonical involution maps $\kappa_M$ and $\kappa_N$. Due to this, we have different projection maps $\tau_{TM}$ and $T\tau_M$ acting on corresponding double tangent spaces onto the tangent space $TM$. The commutator on the tangent space (outer block) yields us the result from \cite{ashutosh} (see Fig. \ref{fig:commutator}), which is the relation between the discretization maps on $M$ and $N$, which are $R_d$ and $R_{d,\phi}$ respectively.

\section{Feedback Linearizable Discretizations for SODEs}
\label{sec:fl_disc}
Let $x \in M$ and $(x, \dot{x}) \in TM$ be the coordinates on the manifold $M$ and the induced coordinates on the tangent bundle of $M$, respectively. We know that a second-order differential equation is a vector field $X$ such that $\tau_{TM}(X) = T \tau_M (X)$. This implies that the vector field $X$ on $TM$ is a section of the second-order tangent bundle $TTM$. Locally, if we take coordinates $(x^i)$ on $M$ and induced coordinates $(x^i, \dot{x}^i)$ on $TM$, then:
\begin{equation}
    X = \dot{x}^i \frac{\partial}{\partial x^i} + X^i(x^i, \dot{x}^i) \frac{\partial}{\partial \dot{x}^i}
\end{equation}
To find the integral curves of $X$ is equivalent to solving the SODE:
\begin{equation}
\label{eq:sode}
    \dfrac{d^2}{dt^2}x(t) = X \left( x(t), \dfrac{d}{dt}x(t) \right) 
\end{equation}
Now, we wish to discretize this using the notion of the discretization map on $TM$. We would like to tangently lift a discretization on $M$ to obtain $R_d^T: TTM \lra TM \times TM$ as defined in Proposition \ref{prop:lift}. This yields the following numerical scheme \cite{21MBLDMdD}:
\begin{equation}
\label{eq:disc}
\begin{split}
    h X \left( \left(\tau_{TM} \circ \left(R_d^T\right)^{-1}\right)(x_k, y_k; x_{k+1}, y_{k+1})\right) \\ = \left(R_d^T\right)^{-1} (x_k, y_k; x_{k+1}, y_{k+1})
\end{split}
\end{equation}

\subsection{Example: Symmetric Discretization}
Let us say we choose the midpoint  discretization on $N={\mathbb R}^n$, denoted by $R_d$ of the following form:
\begin{equation}
    R_d(\tilde{x}, \tilde{y}) = \left(\tilde{x} - \dfrac{\tilde{y}}{2}, \tilde{x} + \dfrac{\tilde{y}}{2} \right)
\end{equation}
for some $(\tilde{x}, \tilde{y}) \in TN$.
From Proposition \ref{prop:lift}, we can find the tangent lift of $R_d$ as follows:
\begin{equation*}
\begin{split}
    D_{(\tilde{x},\tilde{y})}R_d(\tilde{x},\tilde{y}) & = \pmat{
    \text{Id} & -\frac{\text{Id}}{2} \\
    \text{Id} & \frac{\text{Id}}{2}
    } \\
    D_{(\tilde{x},\tilde{y})}R_d(\tilde{x},\tilde{y})(\dot{\tilde{x}}, \dot{\tilde{y}})^T & = \left( \dot{\tilde{x}} - \dfrac{\dot{\tilde{y}}}{2}, \dot{\tilde{x}} + \dfrac{\dot{\tilde{y}}}{2} \right)
\end{split}
\end{equation*}

\begin{equation}
    R_d^T(\tilde{x}, \dot{\tilde{x}}, \tilde{y}, \dot{\tilde{y}}) = \left( \tilde{x} - \dfrac{\tilde{y}}{2}, \tilde{x} + \dfrac{\tilde{y}}{2}, \dot{\tilde{x}} - \dfrac{\dot{\tilde{y}}}{2}, \dot{\tilde{x}} + \dfrac{\dot{\tilde{y}}}{2}\right)
\end{equation}
which is a discretization on $TN$.

Now, to lift $R_d^T$ to obtain $R_{d,\phi}^T$, we use Proposition \ref{prop:lift-commutator}, which gives:

\begin{equation}
    R^T_{d,\phi} = (T \phi \times T \phi)^{-1} \circ R_d^T \circ TT\phi
\end{equation}
 which is also a discretization map on $TM$.

 Using the numerical scheme from Equation \ref{eq:disc}, we obtain:
 \begin{equation}
     \begin{split}
         \dfrac{x_{k+1} - x_k}{h} & = \dfrac{y_{k+1} + y_k}{2}, \\
         \dfrac{y_{k+1} - y_k}{h} & = X \left(\dfrac{x_k + x_{k+1}}{2}, \dfrac{y_k + y_{k+1}}{2} \right)
     \end{split}
 \end{equation}
 which is the numerical scheme for a symmetric discretization of the SODE \ref{eq:sode}.

\subsection{MF-Linearizable discretizations}

We can also apply the discretization of second-order differential equations for controlled mechanical systems. 
\begin{thm}\label{theorem-feed}
Let $R_d$ be a discretization map for the linear mechanical system ($(\mathcal{LMS})_{(n,m)}$) given by  (\ref{sode-lin}) preserving linearity. Then the mechanical control system 
(\ref{SODE-nonlinear}) ($(\mathcal{MS})_{(n,m)})$  admits, using $R_{d,\phi}$, a discretization that is feedback linearizable.
\end{thm}
\begin{proof}
Now if we transform the linearized system (\ref{sode-lin})
\[
            \frac{d^2{\tilde{x}}}{dt^2}  = A \tilde{x} + \sum_{s=1}^m b_s \tilde{u}_s
\]
using a discretization map $R_d$ that preserves linearity and defining $\tilde{u}_k=\tilde{u}(t_{k})$ we obtain as a discretization the 
controlled second–order linear difference equation
\begin{equation}\label{discrete-second}
\tilde{x}_{k+2}= \tilde{A} \tilde{x}_k+\tilde{B} \tilde{x}_{k+1} + \sum_{s=1}^m b_s (\tilde{u_k})_s
\end{equation}
Now since \[
    R^T_{d,\phi} = (T \phi \times T \phi)^{-1} \circ R_d^T \circ TT\phi
\]
then the discrete system 
  \begin{align*}
  &h X \left( \left(\tau_{TM} \circ \left(R_{d,\phi}^T\right)^{-1}\right)(x_k, y_k; x_{k+1}, y_{k+1}), u_k\right) \\ &= \left(R_{d,\phi}^T\right)^{-1} (x_k, y_k; x_{k+1}, y_{k+1})
  \end{align*}
is feedback linearizable to (\ref{discrete-second}) where $X$ is the second order system defined in (\ref{SODE-nonlinear}) and 
\begin{align*}
(\tilde{x}_k, \tilde{y}_k)&=T\phi(x_k, y_k)\\
                    (u_k)_r &= \gamma_{ij}^l(\bar{x}_k) \bar{y}_k^i\bar{y}_k^j + \alpha_l(\bar{x}_k) + \sum_{s=1}^m \beta_s^r(\bar{x}_k) (\tilde{u}_k)_s
\end{align*}
where $(\bar{x}_k, \bar{y}_k)=\tau_M ( R_{d,\phi}^{-1}(x_k, y_k; x_{k+1}, y_{k+1}))$
\end{proof}

\section{Example}
\label{sec:example}

\subsection{Example 1}
\label{subsec:ex-1}
Here, we consider an example: a simple mechanical system - the inertia wheel pendulum.
\begin{figure}[h]
    \centering
    \includegraphics[width=0.3\textwidth]{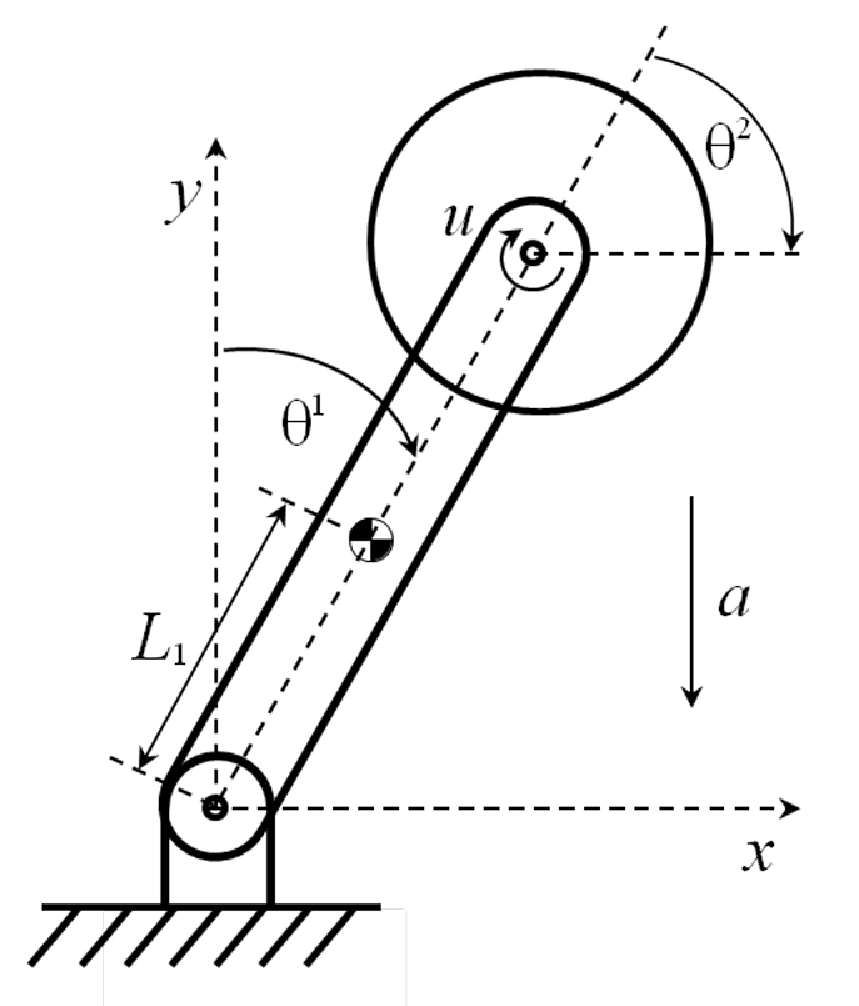}
    \caption{The inertia wheel pendulum}
    \label{fig:inertia-wheel}
\end{figure}
The equations of motion are given by
\begin{equation}
    \begin{split}
        \mathfrak{m}_{11} \Ddot{\theta}^1 + \mathfrak{m}_{12} \Ddot{\theta}^2 + c^1 = 0\\ 
        \mathfrak{m}_{21} \Ddot{\theta}^1 + \mathfrak{m}_{22} \Ddot{\theta}^2 = u
    \end{split}
\end{equation}

where 
\begin{equation*}
    \begin{split}
        \mathfrak{m}_{11} & = m_d + J_2, \
        \mathfrak{m}_{12} = \mathfrak{m}_{21} = \mathfrak{m}_{22}  = J_2 \\
        m_d  & = L_1^2(m_1 + 4 m_2) + J_1, \
        m_0  = a L_1(m_1 + 2m_2) \\
        c^1  & = -m_0 \sin{\theta^1}
    \end{split}
\end{equation*}

Taking $(\theta^1, \theta^2) = (x_1, x_2)$ and correspondingly $(\dot{\theta}^1, \dot{\theta}^2) = (y_1, y_2)$, we get the following equations:
\begin{equation}
\begin{split}
    \dot{x}_1  & = y_1, \
    \dot{x}_2  = y_2 \\
    \dot{y}_1  & = e_1 + g_1 u, \
    \dot{y}_2  = e_2 + g_2 u
\end{split}
\end{equation}
where 
\begin{equation*}
    \begin{split}
        e_1 = \dfrac{m_0}{m_d}\sin{x_1}, & \ \ g_1 = -\dfrac{1}{m_d} \\
        e_2 = -\dfrac{m_0}{m_d} \sin{x_1}, & \ \ g_2 = \dfrac{m_d + J_2}{m_d J_2}
    \end{split}
\end{equation*}

\subsubsection{MF-Linearization}

We will verify $(MD1 - MD3)$ from Proposition \ref{prop:planar_mech}, since the mechanical system here is a planar mechanical system.

First, we calculate:
\begin{equation}
\begin{split}
    \text{ad}_e g & = 0 - \pmat{
    \frac{m_0}{m_d} \cos{x^1} & 0 \\
    - \frac{m_0}{m_d} \cos{x^1} & 0
    }
    \pmat{
    -\frac{1}{m_d} \\
    \frac{m_d + J_2}{m_d J_2}
    } \\
    & = 
    \pmat{
    \frac{m_0}{m_d^2}\cos{x^1} \\
    -\frac{m_0}{m_d^2}\cos{x^1}
    }
\end{split}
\end{equation}

It can be seen that $g$ and $\text{ ad}_e g$ are independent (except at $x_1 = \pm \frac{\pi}{2}$). Thus, $MD1$ is satisfied. To verify $MD2$,
\begin{equation}
    \begin{split}
        \nabla_g g & = \left(\dfrac{\partial g_i}{\partial x_j}g_j + \Gamma^i_{jk}g_j g_k \right) \dfrac{\partial}{\partial x_i} = 0 \in \mathcal{E}^0 \\
        \nabla_{\text{ad}_e g} g & = 0 \in \mathcal{E}^0
    \end{split}
\end{equation}
which is also verified. Lastly, for $MD3$,
\begin{equation}
\label{eq:md3-1}
    \begin{split}
        \nabla^2_{g, \text{ad}_e g} \text{ad}_e g = \nabla^2_{\text{ad}_e g, g} \text{ad}_e g = \pmat{
        \frac{m_0^2}{m_d^5}\cos^2{x_1} \\
        - \frac{m_0^2}{m_d^5} \cos^2{x_1}
        }
    \end{split}
\end{equation}
Thus, we have:
\begin{equation}
    \nabla^2_{g, \text{ad}_e g} \text{ad}_e g - \nabla^2_{\text{ad}_e g, g} \text{ad}_e g = 0 \in \mathcal{E}^0
\end{equation}
Therefore, all the conditions $(MD1 - MD3)$ are satisfied, and the given system is $MF$-Linearizable.

We have the diffeomorphism $\Phi(x,y) = \left( \phi(x), D\phi(x)y\right)$, which is given by:
\begin{equation}
\begin{split}
        \tilde{x}_1  = \dfrac{m_d+J_2}{J_2}x_1 + x_2, \ 
        \tilde{x}_2  = \dfrac{m_0}{J_2}\sin{x_1} \\
        \tilde{y}_1  = \dfrac{m_d+J_2}{J_2}y_1 + y_2,\
        \tilde{y}_2  = \dfrac{m_0}{J_2}\cos{x_1}y_1
\end{split}
\end{equation}

Taking $\tilde{\text{x}} = \pmat{\tilde{x}_1 & \tilde{x}_2 & \tilde{y}_1 & \tilde{y}_2 }^T$, such that the linearized equations become:
\begin{equation}
    \dfrac{d}{dt} \tilde{\text{x}}
    = A \tilde{\text{x}} + C \tilde{u}
\end{equation}

Here, the matrices $A = \pmat { 0 & 0 & 1 & 0 \\ 0 & 0 & 0 & 1 \\ 0 & 1 & 0 & 0 \\ 0 & 0 & 0 & 0 }$, $C = \pmat{0 \\ 0 \\ 0 \\ 1}$ and
$\tilde{u} = \psi (x, y, u)$ is the auxiliary control, such that:
\begin{equation}
    \tilde{u} = -\dfrac{m_0}{J_2} \sin{x_1}y_1^2 + \dfrac{m_0^2}{2m_dJ_2}\sin{2x_1} - \dfrac{m_0}{m_d J_2} \cos{x_1} u
\end{equation}

From Theorem \ref{theorem-feed} using for instance the symmetric map $R_d$, we obtain the discretization $R_{d,\phi}$ and the corresponding linearizable discretization of the initial system: 
\[
\begin{array}{l}
\dfrac{x_{1,k+1}-x_{1,k}}{h}+\dfrac{J_2}{m_d+J_2}\left(\dfrac{x_{2,k+1}-x_{2,k}}{h}-\dfrac{y_{2,k+1}+y_{2,k}}{2}\right) \\ = y_{1,k+1/2}\\
\dfrac{\sin{x_{2,k+1}}-\sin{x_{2,k}}}{h}=\dfrac{y_{1,k}\cos{x_{1,k}}+\sin{x_{2,k}}
y_{1,k}\cos{x_{1,k}}}{2}\\
\dfrac{m_d+J_2}{J_2}\dfrac{y_{1,k+1}-y_{1,k}}{h}+\dfrac{y_{2,k+1}-y_{2,k}}{h}= \dfrac{m_0}{J_2}\sin x_{k+1/2}\\
\dfrac{m_0}{hJ_2}\left(
y_{1,k+1}\cos x_{1,k+1}-y_{1,k}\cos x_{1,k}\right)
=\tilde{u}_k
\end{array}
\]
where
\begin{align*}
    \tilde{u}_k &= -\dfrac{m_0}{J_2} \sin{x_{1,k+1/2}}y_{1,k+1/2}^2 + \dfrac{m_0^2}{2m_dJ_2}\sin{2x_{1,k+1/2}}\\
    &- \dfrac{m_0}{m_d J_2} \cos{x_{1,k+1/2}} u_k
\end{align*}
\subsubsection{Stabilization}
We use the classical pole placement technique to obtain a control gain matrix $K$, such that $\tilde{u} = - K \tilde{x}$.

Let us choose the poles of the closed-loop system to be:
\begin{equation}
    \lambda = -10, -20, -30, -40
\end{equation}

Correspondingly, we obtain 
\begin{equation}
    K = \bmat{240000 & 3500 & 50000 & 100}
\end{equation}

We denote $\text{x} = \bmat{x_1 & x_2 & y_1 & y_2}^T$ to get:
\begin{equation}\label{eq:feedback}
    \dfrac{d}{dt} \text{x}
    = \left( A - C K \right) \text{x}
\end{equation}

\subsubsection{Discretization}
If we have the system as $\dot{\text{x}} = (A-CK)\text{x} = F(\text{x})$. 
Let $h$ denote a (fixed) sampling time and $h' = \frac{h}{2}$. We utilize the symmetric discretization formulated in Section \ref{sec:fl_disc}:
\[
F(\text{x}_k; h/2) = F(\text{x}_{k+1}; -h/2)
\]
$$
\text{x}_k + h' (A-CK) \text{x}_k = \text{x}_{k+1} - h'(A-CK) \text{x}_{k+1} 
$$
\begin{equation}
    \therefore \text{x}_{k+1} = (I - h'(A-CK))^{-1}(I + h'(A-CK)) \text{x}_k
\end{equation}

\subsubsection{Results}
We use the following parameters from \cite{nowicki} and \cite{SPONG20011845}:
\begin{equation}
\begin{split}
    L_1 & = 0.063 \ [m] \\
    m_1 & = 0.02 \ [kg] \\
    m_2 & = 0.3 \ [kg] \\
    J_1 & = 47 \cdot 10^{-6} \ [kg \cdot m^2] \\
    J_2 & = 32 \cdot 10^{-6} \ [kg \cdot m^2] \\
    a & = 9.81 \ [m s^{-2}] \\
    m_0 & = 0.3832 \ [kg \cdot m^2 s^{-2}] \\
    m_d & = 49 \cdot 10^{-4} \ [kg \cdot m^2] 
\end{split}
\end{equation}

The comparison results between the proposed discretization scheme and ODE45 for the system, for a sampling time of $h = 0.01$, and initial conditions $\theta^1(0) = \frac{\pi}{4}, \theta^2(0) = \dot{\theta}^1(0) = \dot{\theta}^2(0) = 0$ are shown in Figure \ref{fig:sim-results}. The  error is plotted in Figure \ref{fig:e1-ed1}.
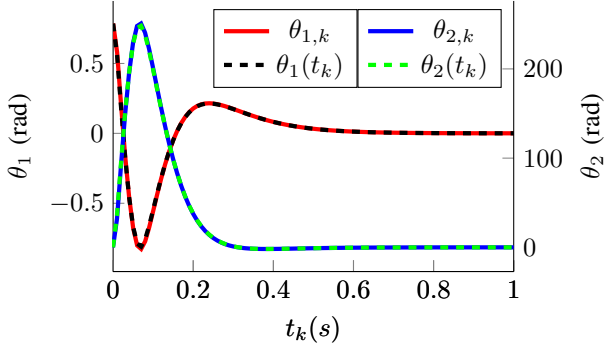
\begin{figure}[ht]
  \begin{center}
    \begin{tikzpicture}
      \begin{axis}[
          width=0.8\linewidth,
          height=0.6\linewidth,
          xlabel=$t_k (s)$,
          ylabel=$\theta_1$ (rad),
          xmin = 0, xmax = 1,
          axis x line* = none,
          axis y line* = left,
          legend style={at={(0.61,0.97)}}
        ]
        \addplot[mark=none, red, ultra thick] table[x=t,y=theta1, col sep=comma]{states_original.csv};
        \addlegendentry{$\theta_{1,k}$}
        \addplot [mark=none, black, dashed, ultra thick]
        table[x=t,y=theta1,col sep=comma]{states_ode45.csv}; 
        \addlegendentry{$\theta_1(t_k)$}
      \end{axis}
    \begin{axis}[
          width=0.8\linewidth,
          height=0.6\linewidth,
          xlabel=$t_k (s)$,
          ylabel=$\theta_2$ (rad),
          xmin = 0, xmax = 1,
          axis x line* = none,
          axis y line* = right,
          legend pos=north east
        ]
        \addplot[mark=none, blue, ultra thick] table[x=t,y=theta2, col sep=comma]{states_original.csv};
        \addlegendentry{$\theta_{2,k}$}
        \addplot [mark=none, green, dashed, ultra thick]
        table[x=t,y=theta2,col sep=comma]{states_ode45.csv};
        \addlegendentry{$\theta_2(t_k)$}
      \end{axis}  
    \end{tikzpicture}
    \caption{System states $x_k$ for symmetric discretization plotted against exact discretization (ODE45) $x(t_k)$ for $t_k \in [0, 1]$}
    \label{fig:sim-results}
  \end{center}
\end{figure}

\begin{figure}[ht]
  \begin{center}
    \begin{tikzpicture}
    \begin{axis}[
          width=0.8\linewidth,
          height=0.6\linewidth,
          xlabel=$t_k (s)$,
          ylabel=$|| e(k) ||$,
          xmin = 0, xmax = 1,
          axis x line* = none,
          axis y line* = left,
          legend pos=north east
        ]
        \addplot[mark=none, blue, thick] table[x=t,y=e1, col sep=comma]{error_norms.csv};
        \addlegendentry{$\lVert \theta_{1,k}-\theta_1(t_k) \rVert$}
      \end{axis}  
      \begin{axis}[
          width=0.8\linewidth,
          height=0.6\linewidth,
          xlabel=$t_k (s)$,
          ylabel=$||\dot{e}(k)||$,
          xmin = 0, xmax = 1,
          axis x line* = none,
          axis y line* = right,
          legend style={at={(0.97,0.77)}}
        ]
        \addplot[mark=none, red, thick] table[x=t,y=ed1, col sep=comma]{error_norms.csv};
        \addlegendentry{$\lVert \dot{\theta}_{1,k}-\dot{\theta}_1(t_k) \rVert$}
      \end{axis}  
    \end{tikzpicture}
    \caption{Magnitude of error norm for $\theta_1$ and $\dot{\theta}_1$}
    \label{fig:e1-ed1}
  \end{center}
\end{figure}
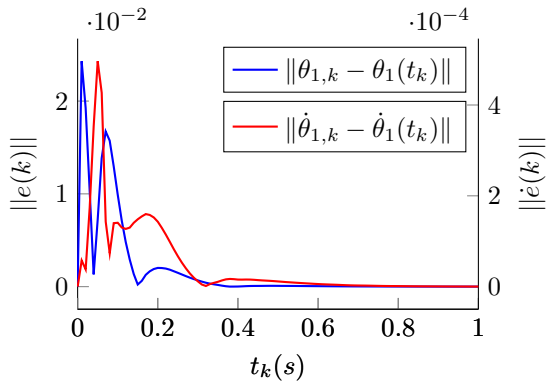

\subsection{Example 2}
Now we consider another example: the rotational dynamics on $M = SO(3)$,  which is described by the following set of equations:
\begin{equation}\label{eq:dyn-rot}
    \begin{split}
        \dot{\textbf{R}}(t) =& \textbf{R}(t)  \widehat{\Omega} (t)  \\
        \dot{\Omega}(t) =& J^{-1} \left( -\Omega(t) \times J\Omega(t) + \tau(t) \right)
    \end{split}
\end{equation}

where $\textbf{R}(t) \in SO(3), \Omega(t) \in \R[3] , \tau(t) \in \R[3]$, 
$$\widehat{\Omega} = \begin{bmatrix}
    0 & -\Omega_3 & \Omega_2 \\ \Omega_3 & 0 & -\Omega_1 \\ -\Omega_2 & \Omega_1 & 0
\end{bmatrix} \in \mathfrak{so}(3) \equiv TM$$
Note that we can simplify the above dynamics \eqref{eq:dyn-rot} by taking the following change of input:

\begin{equation}\label{eq:inp-trans}
    \tau(t) = \Omega(t) \times J\Omega(t) + Ju(t)
\end{equation}
where $u(t) \in \R[3]$ is the new control signal. This yields the modified dynamics on $SO(3)$:
\begin{equation}\label{eq:two-odes-rot}
\begin{split}
    \dot{\textbf{R}}(t) =& \textbf{R}(t) \widehat{\Omega}  (t)\\
    \dot{\Omega}(t) =& u(t)
\end{split}
\end{equation}

This can also be written as an SODE of the form:
\begin{equation}\label{eq:sode-rot}
    \Ddot{\textbf{R}}(t) = \mathbf{R}(t) (\widehat{\Omega}(t))^2 + R \widehat{u}(t)
\end{equation}

We can see that this is of the form $\Ddot{x} = f(x,u)$, which is not a control affine form. We can convert this to a control affine form through a transformation $\mathcal{V} : \R[3 \times 3] \lra \R[9]$ given in the Appendix \ref{appendix-a}. Thus choosing, $x = [\mathcal{V}(\textbf{R}) ,\Omega]^T \in \R[9]$ we get the set of equations:
\begin{equation}
        \dot{x}(t) = f(x(t)) + g(x)u(t)
\end{equation}
where $f(x) = \begin{bmatrix}
    r_{i2}\Omega_3 - r_{i3}\Omega_2 \\
    r_{i3}\Omega_1 - r_{i1}\Omega_3 \\
    r_{i1}\Omega_2 - r_{i2} \Omega_1 \\
    0_{3 \times 1} 
\end{bmatrix}$, $g(x) = \begin{bmatrix}
    0_{9 \times 3} & I_{3 \times 3}
\end{bmatrix}^T$

\subsubsection{MF-Linearization}
In order to prove that the system in Equation \eqref{eq:two-odes-rot} is feedback linearizable while preserving the mechanical structure, we write the actual dynamics on $TSO(3)$. Let $x(t) = R(t) \in SO(3)$. Thus, we have:
\begin{equation}\label{eq:so3-dyn}
    \begin{split}
        \dot{x}(t) =& y(t) \\
        \dot{y}(t) =& y(t) x^T(t) y(t) + x(t) \widehat{u}(t)
    \end{split}
\end{equation}

where $y(t) \in \R[3 \times 3], u(t) \in \R[3]$. Comparing to the general SODE form in \eqref{SODE-nonlinear}, one can observe that $e(x) = 0$ and $g(x)$ are linear functions, i.e. $\pm x_{ij}$. Thus, it is straightforward to prove that the conditions for MF-Linearizability in Theorem \ref{thm:mfl} are satisfied.

Now we consider the following diffeomorphism:
\begin{equation}
    \begin{split}
        \xi(t) = [\widecheck{\text{Log}(R(t))}] \\
        \eta(t) = \Omega(t)
    \end{split}
\end{equation}

where $\widecheck{(\cdot)} : \mathfrak{so}(3) \lra \R[3]$ is the inverse skew-symmetry (or vectorize) operator, $\text{Log} : SO(3) \lra \mathfrak{so}(3)$ is the standard logarithm map, and $\xi, \eta \in \R[3]$. Also, the inverse diffeomorphism yields $R(t) = \exp(\widehat{\xi}(t)), \Omega(t) = \eta(t) $.

Choosing $z(t) = [\xi(t), \eta(t)]^T$ such that $z \in N \equiv \R[6]$, we have the feedback linearized dynamics:
\begin{equation}
    \dot{z}(t) = Az(t) + Bv(t)
\end{equation}
where $A = \begin{bmatrix}
    0_{3 \times 3} & I_{3 \times 3} \\
    0_{3 \times 3} & 0_{3 \times 3}
\end{bmatrix}, B = \begin{bmatrix}
    0_{3\times 3} \\ I_{3 \times 3}
\end{bmatrix}$.

It is interesting to note that the initial nonlinear equations have reduced to a single linear ODE on $\R[6]$ where the feedback linearizing control is $v(t) = \tau(t) - \Omega(t) \times J\Omega(t)$.

\subsubsection{Stabilization}
We again use feedback control to stabilize the system on $N$. We choose $v(t) = -Kz(t)$, where $K = [K_1 \ K_2]$ s.t. $K_1, K_2 \in \R[3 \times 6]$. This yields again, an equation similar to the previous example in \eqref{eq:feedback}.

\subsubsection{Discretization}
Here, let us choose the Forward Euler map for discretization, i.e. $R_d(x, v) = (x, x + v)$. Applying this map to the dynamics on $N$, we get the following discretization scheme:

\[
    z_{k+1} = (I + h(A-BK))z_k
\]
where $h$ is the step-size, $z_k = [\xi_k, \eta_k]^T$.
This yields:
\begin{equation}\label{eq:rot-disc-fl}
\begin{split}
\xi_{k+1} =& \xi_k + h \eta_k \\
\eta_{k+1} =& \eta_k - hK_1\xi_k - hK_2 \eta_k
\end{split}
\end{equation}

Lifting this scheme back, we obtain:
\begin{equation}\label{eq:rot-disc-og}
\begin{split}
    R_{k+1} =& R_k \exp(h\widehat{\Omega}_k) \\
    \Omega_{k+1} =& \Omega_k - hK_1 [\widecheck{\text{Log}(R_k)}] - hK_2\Omega_k
\end{split}
\end{equation}

\subsubsection{Results}

We use the following parameters for the simulations:

\begin{equation}
    \begin{split}
        R(0) &= \bmat{\cos{\pi/2} & 0 & -\sin{\pi/2} \\ 0 & 1 & 0 \\ \sin{\pi/2} & 0 & \cos{\pi/2}}, \Omega(0) = \bmat{0 \\ 0 \\ 0} \\ 
        K_1 &= 5I , K_2 = 10I \\
    \end{split}
\end{equation}

\vspace{-5pt}
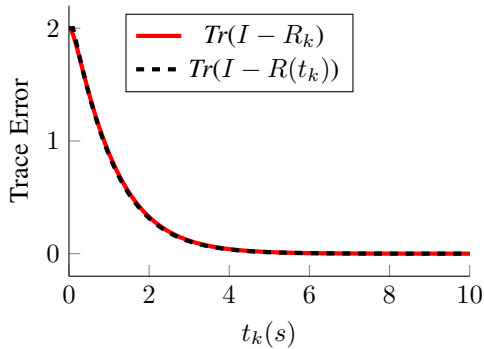
\begin{figure}[ht]
  \begin{center}
    \begin{tikzpicture}
      \begin{axis}[
          width=0.8\linewidth,
          height=0.6\linewidth,
          xlabel=$t_k (s)$,
          ylabel=Trace Error,
          xmin = 0, xmax = 10,
          axis x line* = none,
          axis y line* = left,
          legend style={at={(0.71,0.97)}}
        ]
        \addplot[mark=none, red, ultra thick] table[x=t,y=err, col sep=comma]{rot-dyn.csv};
        \addlegendentry{\textit{Tr}($I - R_k$)}
        \addplot [mark=none, black, dashed, ultra thick]
        table[x=t_k,y=err_k,col sep=comma]{rot-dyn.csv}; 
        \addlegendentry{\textit{Tr}($I - R(t_k)$)}
      \end{axis}
    \end{tikzpicture}
    \vspace{-10pt}
    \caption{Error $e_k$ for forward Euler discretization plotted against exact discretization (ODE45) $e(t_k)$ for $t_k \in [0, 10]$}
    \label{fig:rot-results}
  \end{center}
\end{figure}

\begin{figure}[ht]
  \begin{center}
    \begin{tikzpicture}
      \begin{axis}[
          width=0.8\linewidth,
          height=0.6\linewidth,
          xlabel=$t_k (s)$,
          ylabel=$\omega_k$  (rad/s),
          xmin = 0, xmax = 10,
          axis x line* = none,
          axis y line* = left,
          legend style={at={(0.61,0.97)}}
        ]
        \addplot[mark=none, red, ultra thick] table[x=t,y=p, col sep=comma]{rot-dyn.csv};
        \addlegendentry{$p$}
        \addplot[mark=none, blue, ultra thick] table[x=t,y=q, col sep=comma]{rot-dyn.csv};
        \addlegendentry{$q$}
        \addplot[mark=none, green, thick] table[x=t,y=r, col sep=comma]{rot-dyn.csv};
        \addlegendentry{$r$}
      \end{axis}
    \end{tikzpicture}
    \vspace{-10pt}
    \caption{Angular velocities of the rigid body stabilizing to $[0, 0, 0]^T$ }
    \label{fig:ang-results}
  \end{center}
\end{figure}
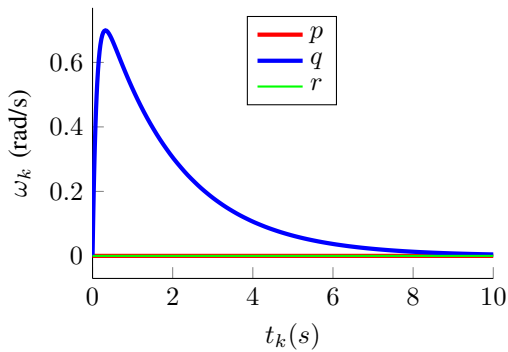

The simulation results for stabilization of rigid body dynamics is shown. Figure \ref{fig:rot-results} shows the error in the rotation matrix trace, while Figure \ref{fig:ang-results} shows the angular velocities.
 
\section{Conclusions}
This paper provides a theoretical basis for further developments in feedback linearizable discretizations of second-order mechanical systems. In a forthcoming paper, we plan to propose a method to functionally compose discretizations to obtain higher-order integrators, using multi-rate sampling, that are feedback linearizable (see \cite{ashutosh}). Moreover, we plan to analyze discrete Sundman transformations in this setting. 


\vspace{-5pt}
\section{Appendix}
\subsection{Control-affine formalization on $SO(3)$}
\label{appendix-a}
In order to translate the dynamics on $\mathfrak{so}(3)$ to the equivalent dynamics on $T\R[n]$ where $n=9$ is chosen appropriately in order to vectorize the equations, we can utilize the diffeomorphism $\mathcal{V} : \R[ 3 \times 3] \lra \R[9]$ such that:

\[
    \begin{bmatrix}
        r_{11} & r_{12} & r_{13} \\
        r_{21} & r_{22} & r_{23} \\ 
        r_{31} & r_{32} & r_{33}
    \end{bmatrix} \mapsto \begin{bmatrix}
        r_{11} \\ r_{12} \\ r_{13} \\ \vdots \\  r_{31} \\ r_{32} \\ r_{33}
    \end{bmatrix}
\]

This transformation also converts the system into a control affine form, as seen in the dynamics and used in \cite{fl-so3}.

\subsection{Mechanical Feedback Linearization}
\label{appendix-b}
\begin{defn}
    Let $MF$ be a group of transformations such that:
    \begin{enumerate}
        \item coordinate transformations given by diffeomorphisms
            \begin{equation}
            \label{eq:mf_diff}
                    \phi: M  \lra N; \
                    x  \mapsto \tilde{x} = \phi(x)
            \end{equation}
        \item mechanical feedback transformations, denoted by $(\alpha, \beta, \gamma)$ such that
            \begin{equation}
                \begin{split}
                    u_r = \gamma_{jk}^r y^jy^k + \alpha_r(x) + \sum_{s=1}^m \beta_s^r(x) \tilde{u}_s
                \end{split}
            \end{equation}
or more compactly,
            \begin{equation}
                u = y^T \gamma y + \alpha + \beta \tilde{u}
            \end{equation}
    \end{enumerate}
\end{defn}

Now, using this group, we define the feedback linearization for mechanical systems:

\begin{defn}
    Two mechanical systems $(\mathcal{MS})_{(n,m)} = (M,\nabla, \mathfrak{g}, e)$ and $(\widetilde{\mathcal{MS}})_{(n,m)} = (N, \widetilde{\nabla}, \mathfrak{\tilde{g}}, \tilde{e})$ are mechanical feedback equivalent if there exists $(\phi, \alpha, \beta, \gamma) \in MF$ such that
    \begin{equation}
        \begin{split}
            \phi: M \lra N \ \ \phi(x) & = \tilde{x}, \
            \phi_*\left(\nabla - \sum_{r=1}^m g_r \otimes \gamma^r\right) = \widetilde{\nabla} \\
            \phi_*\left( \sum_{r=1}^m \beta^r_s g_r \right) & = \tilde{g}_s, \
            \phi_* \left( e + \sum_{r=1}^m g_r \alpha^r \right) = \tilde{e}
        \end{split}
    \end{equation}
\end{defn}

\bibliographystyle{IEEEtran}
\bibliography{IEEEabrv, ref}

\end{document}